\documentclass[12pt,a4paper]{article}
\usepackage{amsthm,amsfonts,amsmath,amssymb,bookmark}

\theoremstyle{plain}
\newtheorem{theorem}{Theorem}%[section]
\newtheorem{lemma}[theorem]{Lemma}

\theoremstyle{definition}

\newtheorem{remark}[theorem]{Remark}

\theoremstyle{remark}

\def\bq{\begin{eqnarray}}
\def\eq{\end{eqnarray}}
\def\bqq{\begin{eqnarray*}}
\def\eqq{\end{eqnarray*}}
\def\nn{\nonumber}

\def\eps{\varepsilon}
\def\wto{\rightharpoonup}
\def \ess {\rm ess}
\def \supp {\rm supp}

\def\R{\mathbb{R}}

\def\cE {\mathcal{E}}
\def \d {{\rm d}}

\title{Blow-up for biharmonic Schr\"odinger equation with critical nonlinearity}

\author{Thanh Viet Phan \\
\normalsize{Applied Analysis Research Group, Faculty of Mathematics and Statistics,} \\
\normalsize{Ton Duc Thang University, Ho Chi Minh City, Vietnam}\\
\normalsize{phanthanhviet@tdt.edu.vn} 
}

\date{\normalsize\today}

\begin{document}

\maketitle

%%%%%%%%%%%%%%%%%%%%%%%%%%%%%%%%%
%%%%%%%%%%%%%%%%%%%%%%%%%%%%%%%%%

\begin{abstract} We consider the minimizers for the biharmonic nonlinear Schr\"odinger functional
$$
\mathcal{E}_a(u)=\int_{\mathbb{R}^d} |\Delta u(x)|^2 \d x + \int_{\mathbb{R}^d} V(x) |u(x)|^2 \d x  -  a \int_{\mathbb{R}^d} |u(x)|^{q} \d x
$$
with the mass constraint $\int |u|^2=1$. We focus on the special power $q=2(1+4/d)$, which makes the nonlinear term $\int |u|^q$ scales similarly to the biharmonic term $\int |\Delta u|^2$. Our main results are the existence and blow-up behavior of the minimizers when $a$ tends to a critical value $a^*$, which is the optimal constant in a Gagliardo--Nirenberg interpolation inequality.

\bigskip
   
\end{abstract}

%%%%%%%%%%%%%%%%%%%%%%%%%%%%%%%%
%%%%%%%%%%%%%%%%%%%%%%%%%%%%%%%%

\section{Introduction}

In this paper, we consider the existence and behavior of the minimizer for the biharmonic nonlinear Schr\"odinger functional
$$
\cE_a(u)=\int_{\R^d} |\Delta u(x)|^2 \d x + \int_{\R^d} V(x) |u(x)|^2 \d x  -  a \int_{\R^d} |u(x)|^{q} \d x
$$
under the mass constraint
$$
\int_{\R^d} |u(x)|^2 \d x = 1.
$$
The biharmonic operator $\Delta^2$ describes the effects of higher-order
dispersion in nonlinear physics, see e.g. \cite{Karpman-96,BKS-00,FIP-02,Pausader-09} for a detailed discussion of the motivation. The function $V$ stands for an external potential and the parameter $a>0$ stands for the strength of the attraction of the system. 

We are interested in the special power
$$
q=2\Big( 1+ \frac{4}{d}\Big),
$$
which makes the nonlinear term $\int |u|^{q}$ scales in the same way as the biharmonic term $\int  |\Delta u|^2 $. Indeed, by defining $u_\ell(x)=\ell^{d/2} u(\ell x)$, we have $\|u_\ell\|_{L^2}=\|u\|_{L^2}$ and
\bq \label{eq:scale}
\cE_a(u_\ell)= \ell^4 \Big( \int_{\R^d} |\Delta u(x)|^2 \d x  - a \int_{\R^2} |u(x)|^{q} \d x \Big) + \int_{\R^d} V(x/\ell) |u(x)|^2 \d x, \quad \forall \ell>0.
\eq

In case $V=0$, it follows from \eqref{eq:scale} that the functional $\cE_a(u)$ is bounded from below if and only if $a\le a^*$, where $a^*$ is the optimal constant in the Gagliardo--Nirenberg interpolation inequality \cite{Nirenberg-59}
\bq \label{eq:inte-ineq}
\Big( \int_{\R^d} |\Delta u(x)|^2 \d x \Big) \Big( \int_{\R^d} |u(x)|^2 \d x \Big)^{\frac{q-2}{2}} \ge a^* \int_{\R^d} |u(x)|^{q} \d x, \quad \forall u\in H^2(\R^d).
\eq
It has been known that the inequality \eqref{eq:inte-ineq} has a minimizer which can be chosen to be radially symmetric; see \cite[Appendix A]{BL-16}. Moreover, any minimizer $u$, up to a dilation $u\mapsto \mu u(\lambda )$ if necessary, satisfies the Euler-Lagrange equation
\bq  \label{eq:inte-ineq-eq}
(\Delta)^2 u + u  - u^{q-1} =0
\eq
for some constant $\mu \in \R$ (Lagrange multiplier). The uniqueness for solutions to \eqref{eq:inte-ineq-eq} (up to translations) and the uniqueness for minimizers of  \eqref{eq:inte-ineq} both remain open. 

In the present paper, we are interested in the existence and behavior of minimizers for the functional $\cE_a(u)$ when $a\uparrow a^*$, with the presence of the external potential $V$.  Our main results are  

\begin{theorem}[Existence and non-existence]\label{thm:1} Let $V\in L^1_{\rm loc}(\R^d,\R)$ satisfy either
\begin{itemize}
\item [\rm (V1)] $V\ge 0$ and $V(x)\to \infty$ as $|x|\to \infty$; or

\item [\rm (V2)] ${\ess \inf}\, V<0$ and 
$$\min\{V,0\}\in L^{p_1}(\R^d)+L^{p_2}(\R^d), \quad \max\{1,d/4\}<p_1<p_2<\infty.$$
\end{itemize}
Then there exists a constant $a_*\in (0,a^*)$ such that for all $a_*<a<a^*$ the variational problem
\bq \label{eq:Ea}
 E_a=\inf \left\{ \cE_a(u)\,|\, u\in H^2(\R^d), V|u|^2 \in L^1(\R^d), \int |u|^2 =1 \right\}
 \eq
has (at least) a minimizer. On the other hand, if $a \ge a^*$, then $E_a$ has no minimizer. 
\end{theorem}

\begin{theorem}[Blow-up]\label{thm:2} Let $V$ be as in Theorem \ref{thm:1}. Let $a_n \uparrow a^*$ and let $u_n$ be a minimizer for the variational problem $E_{a_n}$ in \eqref{eq:Ea}. Then $\{u_n\}$ blows up when $n\to \infty$, i.e.
$$
\lim_{n\to \infty} \|\Delta u_n \|_{L^2}=\infty.
$$
Moreover, up to a subsequence of $\{u_n\}$, there exist a sequence $\eps_n\to 0^+$, a sequence $\{x_n\}\subset \R^d$ and a minimizer $Q$ for the Gagliardo--Nirenberg interpolation inequality \eqref{eq:inte-ineq}, which satisfies 
$$\|\Delta Q\|_{L^2}=\|Q\|_{L^2}=a^*\int_{\R^d} |Q|^q=1,$$
such that 
$$
\lim_{n\to \infty} \, \eps_n^{d/2} u_n \Big(\eps_n (x+x_n)\Big) = Q(x)\quad \text{strongly in~}H^2(\R^d).
$$
\end{theorem}

Our work is motivated by recent works \cite{GuoSei-14,DenGuoLu-15,GuoZenZho-16} on the existence and blow-up behavior of the minimizers of the Gross-Pitaevskii functional in two dimensions. However, our problem is more difficult than the Gross-Pitaevskii model in many aspects. In particular, the blow-up result is more difficult because of the lack of the uniqueness result for the minimizers of \eqref{eq:inte-ineq} and because of the generality of the external potential $V$. Therefore, we need to use several new tools and ideas. 

We will prove Theorems \ref{thm:1} and \ref{thm:2}  in Section \ref{sec:thm1} and \ref{sec:thm2}, respectively. 

\section{Proof of Theorem \ref{thm:1}} \label{sec:thm1}

As a preliminary step, we have

\begin{lemma} \label{lem:ea->0} $\limsup_{a\uparrow a^*}E_a \le {\ess\inf}\, V.$
\end{lemma}

\begin{proof} Let $Q$ be a radial minimizer for the Gagliardo--Nirenberg interpolation inequality \eqref{eq:inte-ineq}. By a standard method, we can show that $Q$ decays sufficiently fast, see e.g. \cite{DL-07}. By modifying $Q$, for every $\eps>0$ sufficiently small, we can construct a function $\varphi_\eps \in C^\infty_c(\R^d)$ (a smooth function with compact support) such that
\bq \label{eq:density}
{\supp}\, \varphi_\eps \subset B(0, \eps^{-1/6}), \quad \int_{\R^d} |\varphi_\eps|^2 =1 \ge \int_{\R^d}|\Delta \varphi_\eps|^2, \quad a^* \int_{\R^d} |\varphi_\eps|^q \ge 1-\eps. 
\eq
For every $\ell>0$ and $x_0\in \R^d$, we consider the trial state
$$
u(x)=\ell^{d/2} \varphi_\eps (\ell(x-x_0)).
$$
Then rescaling as in \eqref{eq:scale} and using \eqref{eq:density}, we have
\begin{align} \label{eq:trial}
E_a \le \cE_a(u) &= \ell^4 \Big( \int_{\R^d}|\Delta \varphi_\eps|^2 - a \int_{\R^d} |\varphi_\eps|^{q} \Big) + \int_{\R^d} V(x/\ell+x_0)  |\varphi_\eps|^2 \d x \nn\\
&\le \ell^4 \Big(  1- \frac{a}{a^*} (1-\eps) \Big) + {\ess\sup} \{ V(x/\ell+x_0) \,|; |x|\le \eps^{-1/6} \}.
\end{align}
We can choose 
$$ \eps=1-\frac{a}{a^*}, \quad \ell= \eps^{-1/5}$$
to have
$$
\ell^4 \Big(  1- \frac{a}{a^*} (1-\eps) \Big) \to 0, \quad \eps^{-1/6}/\ell \to 0
$$
as $a\uparrow a^*$. Therefore, we deduce from \eqref{eq:trial} that
$$
\limsup_{a\uparrow a^*}E_a \le V(x_0)
$$
for a.e. $x_0\in \R^d$. This ends the proof. 
\end{proof}

Now we come to the non-existence part of Theorem \ref{thm:1}.

\begin{lemma} For every $V\in L^1_{\rm loc}(\R^d,\R)$, we have 
\begin{itemize}

\item $E_a=-\infty$ if $a>a^*$;

\item $E_{a^*}={\ess \inf} \, V$ but it has no minimizer except when $V$ is a constant. 
\end{itemize}
\end{lemma}

\begin{proof} First, we consider the case when $a>a^*$.  In this case, the fact $E_{a}=-\infty$ follows from the estimate \eqref{eq:trial} by choosing $\eps>0$ sufficiently small such that
$$ 1-\frac{a}{a^*}(1-\eps)<0$$
and then taking $\ell\to \infty$.

When $a=a^*$, from Lemma \ref{lem:ea->0} and the monotonicity of $a\mapsto E_{a}$, we find that 
$$E_{a^*}\le {\ess \inf} V.$$
On the other hand, by the Gagliardo--Nirenberg inequality \eqref{eq:inte-ineq} and the trivial inequality
\bq \label{eq:istg}
\int_{\R^d} V|u|^2 \ge \Big( {\ess \inf} V \Big) \int_{\R^d} |u|^2
\eq
we get 
$$
E_{a^*}\ge {\ess \inf} V. 
$$
Thus $E_{a^*}= {\ess \inf} V$. However, if $V \not\equiv {\rm constant}$, then $E_{a^*}$ has no minimizer because the inequality \eqref{eq:istg} is strict when $u$ is a minimizer for the interpolation inequality \eqref{eq:inte-ineq}. The latter claim is a consequence of the known fact that any minimizer of \eqref{eq:inte-ineq} does not vanish in a set of positive measure, which can be deduced using the Euler-Lagrange equation \eqref{eq:inte-ineq-eq}, see e.g. \cite{DL-07,BL-16}.
\end{proof}

Now we come to the existence part of Theorem \ref{thm:1}. The proof is divided into two cases.

\begin{lemma} Let $V$ satisfy condition (V1) in Theorem \ref{thm:1}. Then $E_a$ has a minimizer for all $a\in (0,a^*)$.
\end{lemma}

\begin{proof} Let $\{u_n\}$ be a minimizing sequence for $E_a$. Using $V\ge 0$ and the Gagliardo--Nirenberg inequality \eqref{eq:inte-ineq} we find that both $\|u_n\|_{H^2}(\R^d)$ and $\int_{\R^d} V|u_n|^2$ are bounded. Using Sobolev's embedding  and the fact that $V(x)\to \infty$ as $|x|\to \infty$, we conclude that up to a subsequence, $u_n$ converges to a function $u$ weakly in $H^2(\R^d)$ and strongly in $L^p(\R^d)$ for all $p\in [2,4^*)$ where $4^*$ is the critical power of the Sobolev embedding of $H^2(\R^d)$, i.e. $4^*=+\infty$ if $d\le 4$ and $4^*=2d/(d-4)$ if $d\ge 5$. 

In particular, when $n\to \infty$ we have
$$
1=\int_{\R^d} |u_n|^2 \to \int_{\R^d} |u|^2
$$
and
$$ \int_{\R^d} |u_n|^q \to \int_{\R^d} |u|^q.$$
Moreover, 
$$
\int_{\R^d} |\Delta u_n|^2 \ge  \int_{\R^d} |\Delta u|^2 + o(1)_{n\to \infty},
$$
since $\Delta u_n \wto \Delta u$ weakly in $L^2(\R^d)$ and 
$$
\int_{\R^d} V |u_n|^2 \ge \int_{\R^d} V |u|^2 + o(1)_{n\to \infty}
$$
by Fatou's lemma (the strong convergence $u_n\to u$ in $L^p(\R^d)$ implies that, up to a subsequence, $u_n\to u$ pointwise).

Thus
$$
E_a=\lim_{n\to \infty} \cE_a(u_n)\ge \cE_a(u).
$$
Therefore, $u$ is a minimizer for $E_a$. 
\end{proof}
\begin{remark} From the above proof we also conclude that 
$$
\limsup_{n\to\infty}\int_{\R^d} |\Delta u_n|^2 \le  \int_{\R^d} |\Delta u|^2
$$
and hence $\Delta u_n\to \Delta u$ strongly in $L^2(\R^d)$.

Furthermore, we denote the Fourier transforms of $f$ as $\widehat f$. Since $\widehat {D^\alpha f}(x)=(2\pi i x)^\alpha \widehat {f} (x)$ and by Plancherel's Theorem, we have  
\bqq && {\left\| {\Delta \left( {u - {u_n}} \right)} \right\|^2_{{L^2}}} + (4\pi^2)^2{\left\| {u - {u_n}} \right\|^2_{{L^2}}}\\
 &=& \left\|\widehat { {\Delta \left( {u - {u_n}} \right)} }\right\|^2_{{L^2}}+(4\pi^2)^2\left\| {\widehat {u - {u_n}}} \right\|^2_{{L^2}}\\
 &=&\left\|4\pi^2|x|^2\widehat { { \left( {u - {u_n}} \right)} }\right\|^2_{{L^2}}+(4\pi^2)^2\left\| {\widehat {u - {u_n}}} \right\|^2_{{L^2}}\\
 &=&(4\pi^2)^2 \left\|\sqrt{1+|x|^4}\widehat { { \left( {u - {u_n}} \right)} }\right\|^2_{{L^2}}\\
 &\ge & C \left\|\widehat { { D^\alpha\left( {u - {u_n}} \right)} }\right\|^2_{{L^2}}\\
&=&C \left\| { { D^\alpha\left( {u - {u_n}} \right)} }\right\|^2_{{L^2}}, \forall 1\le |\alpha|\le 2,
  \eqq
  for a constant $C$.
  
  Therefore $ u_n\to u$ strongly in $H^2(\R^d)$.

\end{remark}

\begin{lemma} Let $V$ satisfy condition (V2) in Theorem \ref{thm:1}. Then there exists a constant $a_*\in (0,a^*)$ such that $E_a$ has a minimizer for all $a\in (a_*,a^*)$.
\end{lemma}

\begin{proof} From Lemma \ref{lem:ea->0} and condition (V2), we have
$$
\lim_{a\uparrow a^*} E_a = {\ess \inf}\, V <0.
$$
Therefore, we can find $a_*\in (0,a^*)$ such that 
\bq \label{eq:Ea<0}
E_a<0, \quad \forall a\in (a_*,a^*).
\eq
Now let us prove that $E_a$ has a minimizer for all $a\in (a_*,a^*)$. \\

{\bf Step 1}. First, let us prove that $E_a>-\infty$. For every $\eps>0$, using the assumption (V2), we can show that for all $\eps>0$,
\bq \label{eq:Sobolev}
\eps \int_{\R^d} |\Delta u|^2 + \int V |u|^2 \ge -C_\eps>-\infty, \quad \forall u\in H^2(\R^d), \int |u|^2=1.
\eq
To prove \eqref{eq:Sobolev}, we observe that any function $f\in L^p(\R^d)$ can be decomposed into 
$$f=f_L+g_L, \quad f_L= f \chi (|f|\ge L) , \quad g_L= f \chi (|f|<L) $$
where $\chi (|f|\ge L)$ is the characteristic function of the set $\{x\in \R^d: f(x)\ge L\}$. By Lebesgue Dominated Convergence,
$$
\lim_{L\to \infty} \|f_L\|_{L^p}=0.
$$
Therefore, for every $\eps>0$ we can choose $L$ sufficiently large such that
$$ \|f_L\|_{L^p} \le \eps, \quad \|g_L\|_{L^\infty}<\infty.$$
From this observation and the assumption (V2), we can decompose
$$
\min\{V,0\}=V_1+V_2 + V_3
$$
where
\bq \label{eq:dec-1}
\|V_1\|_{L^{p_1}}\le \eps, \quad \|V_2\|_{L^{p_2}} \le \eps, \quad \|V_3\|_{L^\infty}<\infty.
\eq
On the other hand, using Sobolev's embedding $H^2(\R^d)\subset L^{s}(\R^d)$ for all $s\in [2,4^*)$ (recall that $4^*=+\infty$ if $d\le 4$ and $4^*=2d/(d-4)$ if $d\ge 5$) and H\"older's inequality, we obtain
\bq \label{eq:dec-2}
\int_{\R^d} V|u|^2 \ge \sum_{i=1}^3 \int_{\R^d} V_i |u|^2 \ge - C \Big( \|V_1\|_{L^{p_1}} + \|V_2\|_{L^{p_2}} \Big) \|u\|_{H^2}^2 - \|V_3\|_{L^\infty} \int_{\R^d} |u|^2.
\eq
for a constant $C$ independent of $V$ and $u$. Inserting \eqref{eq:dec-1} into \eqref{eq:dec-2}, we obtain
$$
\int_{\R^d} V|u|^2 \ge -C \eps  \|u\|_{H^2}^2 - C_\eps \int_{\R^d} |u|^2
$$
for all $u\in H^2$ and for all $\eps>0$, where $C_\eps$ is a finite constant depending on $\eps>0$. This inequality is equivalent to \eqref{eq:Sobolev}. 

From \eqref{eq:Sobolev} and the Gagliardo--Nirenberg inequality \eqref{eq:inte-ineq}, we find that for all $\eps>0$,
\bq \label{eq:E-lb}
\cE_a(u) \ge \Big( 1-\frac{a}{a^*}-\eps\Big) \int_{\R^d} |\Delta u|^2 - C_\eps, \quad \forall u\in H^2(\R^d), \int |u|^2=1.
\eq
Of course, we can choose $\eps>0$ sufficiently small such that 
$$  1-\frac{a}{a^*}-\eps >0.$$
We then conclude from \eqref{eq:E-lb} that $E_a>-\infty.$\\

{\bf Step 2.} Let $\{u_n\}$ be a minimizing sequence for $E_a$. From \eqref{eq:E-lb} we obtain that $\{u_n\}$ is bounded in $H^2(\R^d)$. By Sobolev's embedding, up to a subsequence, $u_n$ converges to a function $u$ weakly in $H^2(\R^d)$ and pointwise.  

Now let us pass $n\to \infty$ in the energy functional $\cE_a(u_n)$. First, we have
\bq \label{eq:plm-1}
\int_{\R^d} V |u_n|^2 \ge \int_{\R^d}  V |u|^2 + o(1)_{n\to \infty}.
\eq
Indeed, by Fatou's lemma,
$$
\int_{\R^d} \max\{V,0\} |u_n|^2 \ge \int_{\R^d} \max\{V,0\} |u|^2 + o(1)_{n\to \infty}.
$$
Moreover, since $u_n \wto u$ weakly in $L^p(\R^d)$ for all $p\in [2,4^*)$ (by Sobolev's embedding $H^2(\R^d)\subset L^p(\R^d)$) and the condition in (V2),
$$\min\{V,0\}\in L^{p_1}(\R^d)+L^{p_2}(\R^d), \quad \max\{1,d/4\}<p_1<p_2<\infty,$$
we get
$$
\int_{\R^d} \min\{V,0\} |u_n|^2 = \int_{\R^d} \min \{V,0\} |u|^2 + o(1)_{n\to \infty}.
$$
Thus \eqref{eq:plm-1} holds true.

Next, since $\Delta u_n \wto \Delta u$ weakly in $L^2(\R^d)$, we have
$$
\int_{\R^d} \Big( \Delta u_n - \Delta u \Big) {\Delta u} = o(1)_{n\to \infty}.
$$
This allows us to decompose
\bq \label{eq:plm-2}
\int_{\R^d} |\Delta u_n|^2 = \int_{\R^d} |\Delta u|^2 +  \int_{\R^d} |\Delta (u_n-u)|^2 + o(1)_{n\to \infty}.
\eq
For the nonlinear term, using the pointwise convergence $u_n\wto u$ and Brezis-Lieb's refinement of Fatou's lemma \cite{BreLie-83}, we obtain
\bq \label{eq:plm-3}
\int_{\R^d} |u_n|^q = \int_{\R^d} |u|^q +  \int_{\R^d} |u_n-u|^q + o(1)_{n\to \infty}.
\eq

Putting all estimates \eqref{eq:plm-1}-\eqref{eq:plm-2}-\eqref{eq:plm-3} together, we find that
\bq \label{eq:plm-4}
\cE_a(u_n) \ge \cE_a(u)+ \int_{\R^d} |\Delta (u_n-u)|^2 - a \int_{\R^d} |u_n-u|^q + o(1)_{n\to \infty}.
\eq 

Since $u_n\wto u$ weakly in $L^2(\R^d)$, similarly to \eqref{eq:plm-2} we have
\bq \label{eq:plm-norm}
1= \int_{\R^d} |u_n|^2 = \int_{\R^d} |u|^2 + \int_{\R^d} |u_n-u|^2 + o(1)_{n\to \infty}.
\eq
Thus $\|u_n-u\|_{L^2}\le 1 + o(1)_{n\to \infty}$, and hence 
$$\|u_n-u\|_{L^2}^2 \le \frac{a^*}{a}$$
if $n$ is sufficiently large. Consequently, by the Gagliardo--Nirenberg inequality \eqref{eq:inte-ineq}, 
$$
 \int_{\R^d} |\Delta (u_n-u)|^2 \ge \frac{a^*} {\|u_n-u\|_{L^2}^2} \int_{\R^d} |u_n-u|^q \ge a \int_{\R^d} |u_n-u|^q,
$$
if $n$ is sufficiently large. Thus \eqref{eq:plm-4} implies that 
$$
\cE_a(u_n) \ge \cE_a(u) + o(1)_{n\to \infty}
$$
and hence
\bq \label{eq:plm-final}
E_a \ge  \cE_a(u).
\eq

{\bf Step 3.} To conclude, we need to show that $\|u\|_{L^2}=1$. From  \eqref{eq:plm-norm}, we have known that $\|u\|_{L^2}\le 1$. Moreover, $u\not\equiv 0$ because 
$$\cE_a(u)\le E_a<0$$
due to \eqref{eq:plm-final} and \eqref{eq:Ea<0}. Thus we can normalize
$$ v:=\frac{u}{\|u\|_{L^2}}$$
and write
$$
\cE_a(u) = \|u\|_{L^2}^2 \cE_a(v) + \Big( \|u\|_{L^2}^2 - \|u\|_{L^2}^q \Big) a \int |v|^q. 
$$
Using
$$\cE_a(v)\ge E_a\ge \cE_a(u)$$
we find that
$$
E_a \ge \|u\|_{L^2}^2 E_a  + \Big( \|u\|_{L^2}^2 - \|u\|_{L^2}^q \Big) a \int |v|^q. 
$$
Since  $\|u\|_{L^2}\le 1$, $E_a<0$ and 
$$q=2\Big( 1+ \frac{4}{d}\Big) > 2$$
we then conclude that $ \|u\|_{L^2}=1$. Thus $u$ is a minimizer for $E_a$. 
\end{proof}

\section{Proof of Theorem \ref{thm:2}} \label{sec:thm2}

\begin{proof} Let $a_n \uparrow a^*$ and let $u_n$ be a minimizer for $E_{a_n}$. \\

{\bf Step 1.} First, we show that
\bq \label{eq:bl1}
\lim_{n\to \infty}\|\Delta u_n\|_{L^2}= \infty.
\eq
We assume by contradiction that, up to a subsequence, $\{u_n\}$ is  bounded in $H^2(\R^d)$. Then, up to a subsequence again, we can assume that $u_n$ converges to a function $u$ weakly in $H^2(\R^d)$ and pointwise. By following the proof of Theorem \ref{thm:1}, we obtain that $u$ is a minimizer for $E_{a^*}$. However, this contradicts to the fact that $E_{a^*}$ has no minimizer. Thus \eqref{eq:bl1} holds true.

{\bf Step 2.} Define
$$
\eps_n := \|\Delta u_n\|_{L^2}^{-1/2} \quad \text{and}\quad w_n(x):=\eps_n^{d/2} u_n(\eps_n x).
$$
Then we have
\bq \label{eq:Dun-1}
\|\Delta w_n\|_{L^2} = \|w_n\|_{L^2}=1, \quad \forall n\in \mathbb{N}.
\eq
We will show that $\{w_n\}$ is a minimizing sequence for the Gagliardo--Nirenberg inequality \eqref{eq:inte-ineq}, i.e.
\bq \label{eq:Dun-2}
\lim_{n\to \infty}a^* \int_{\R^d} |w_n|^q = 1.
\eq
Of course, by \eqref{eq:inte-ineq}, we have immediately the upper bound
$$
\limsup_{n\to \infty}a^* \int_{\R^d} |w_n|^q \le \|\Delta w_n\|_{L^2}^2=1.
$$
It remains to prove the lower bound. Recall that when $V$ satisfies either (V1) or (V2), we have Sobolev-type inequality \eqref{eq:Sobolev}, and hence
$$
\eps \int_{\R^d} |\Delta u_n|^2 + \int V |u_n|^2 \ge -C_\eps>\infty, \quad \forall \eps>0.
$$
Therefore,
\begin{align*}
\cE_{a_n}(u_n) &\ge (1-\eps) \int_{\R^d} |\Delta u_n|^2 -a^*\int_{\R^d} |u_n|^q - C_\eps \\
&= \eps_n^{-4} \Big( (1-\eps) \int_{\R^d} |\Delta w_n|^2 - a^* \int_{\R^d} |w_n|^q \Big) - C_\eps, \quad \forall \eps>0.
\end{align*}
On the other hand, by Lemma \ref{lem:ea->0},
$$
\lim_{n\to \infty} \cE_{a_n}(u_n)  = \lim_{n\to \infty} E_{a_n} =  {\ess \inf}\, V < \infty.
$$
Consequently, 
$$
\limsup_{n\to \infty} \eps_n^{-4} \Big( (1-\eps) \int_{\R^d} |\Delta w_n|^2 -a^* \int_{\R^d} |w_n|^q \Big) <\infty, \quad \forall \eps>0.
$$
Since $\eps_n\to 0$ and $\|\Delta w_n\|^2=1$, we obtain
$$
\liminf_{n\to \infty}  a^* \int_{\R^d} |w_n|^q  \ge 1.
$$
Thus \eqref{eq:Dun-2} holds true.\\

{\bf Step 3.} Now we use \eqref{eq:Dun-1} and \eqref{eq:Dun-2} to prove that, up to subsequences and translations, $w_n$ converges strongly in $H^2(\R^d)$ to a minimizer for the Gagliardo--Nirenberg inequality \eqref{eq:inte-ineq}. We will need two useful tools taken from  \cite[Lemma 2.1]{FLL-86} and  \cite[Lemma 6]{Lieb-83}. 

\begin{lemma} \label{lem:tool1} Let $p_1<p_2<p_3$. Let $\{f_n\}$ be a bounded sequence in $L^{p_1}(\R^d)\cap L^{p_3}(\R^d)$ such that
$$\liminf_{n\to \infty} \|f_n\|_{L^{p_2}} >0 .$$
Then there exists $\eta>0$ such that
\bq \label{eq:fn-dd}
 \liminf_{n\to \infty} \Big| \{x\in \R^d\,|\,|f_n(x)|\ge \eta\} \Big| >0.
 \eq
\end{lemma}

\begin{lemma} \label{lem:tool2} Let $\{f_n\}$ be a bounded sequence in $H^1(\R^d)$ such that \eqref{eq:fn-dd} holds for some $\eta>0$. Then up to a subsequence of $\{f_n\}$, there exist a sequence $\{x_n\}\subset \R^d$ and $0\not\equiv f \in H^1(\R^d)$ such that 
$$f_n(.+x_n)\to f \quad \text{weakly in~}H^1(\R^d).$$
\end{lemma}

Let us come back to our problem. From \eqref{eq:Dun-1}, $w_n$ is bounded in $H^2(\R^d)$. Therefore, by Sobolev's embedding, $w_n$ is bounded in $L^2(\R^d)\cap L^{p}(\R^d)$ for some $p>q>2$. Therefore,  \eqref{eq:Dun-2} allows us to use \label{lem:tool1} to find a constant $\eta>0$ such that  
$$
 \liminf_{n\to \infty} \Big| \{x\in \R^d\,|\,|w_n(x)|\ge \eta\} \Big| >0.
$$
Next, applying Lemma \eqref{lem:tool2}, up to a subsequence of $\{w_n\}$, there exist a sequence $\{x_n\}\subset \R^d$ and $0\not\equiv w \in H^1(\R^d)$ such that $w_n(.+x_n)\to w$ weakly in $H^1(\R^d)$. Since $w_n(.+x_n)$ is also bounded in $H^2(\R^d)$, up to a subsequence, we can assume that 
$$w_n(.+x_n) \to w\quad \text{ weakly in $H^2(\R^d)$ and pointwise.} $$

Now we prove that $\|w\|_{L^2}=1$ and it is a minimizer  for the Gagliardo--Nirenberg inequality \eqref{eq:inte-ineq}. We will proceed similarly to the proof of Theorem \ref{thm:1}. To be precise, since $w_n(.+x_n) \to w$ weakly in $H^2(\R^d)$, we have 
$$
\int_{\R^d} \Big( \Delta w_n (x+x_n) - \Delta w(x) \Big) {\Delta w(x)} \d x \to 0,
$$
and hence
$$
\int_{\R^d} |\Delta w_n|^2 = \int_{\R^d} |\Delta w_n(.+x_n)|^2 =  \int_{\R^d} |\Delta w|^2 +  \int_{\R^d} |\Delta (w_n(.+x_n)-w)|^2 + o(1)_{n\to \infty}.
$$
Moreover, since $w_n(.+x_n) \to w$ pointwise, by Brezis-Lieb's refinement of Fatou's lemma \cite{BreLie-83}, 
$$
\int_{\R^d} |w_n|^q = \int_{\R^d} |w_n(.+x_n)|^q = \int_{\R^d} |w|^q +  \int_{\R^d} |w_n(.+x_n)-w|^q + o(1)_{n\to \infty}.
$$
Therefore,
\begin{align} \label{eq:final-1}
\int_{\R^d} |\Delta w_n|^2  - a^* \int_{\R^d} |w_n|^q &= \Big( \int_{\R^d} |\Delta w|^2 - a^* \int_{\R^d} |w|^q \Big) \nn\\
& + \Big( \int_{\R^d} |\Delta (w_n(.+x_n)-w)|^2 - a^* \int_{\R^d} |w_n(.+x_n)-w|^q \Big) \nn\\
&+  o(1)_{n\to \infty}.
\end{align}
For the left side of \eqref{eq:final-1}, from \eqref{eq:Dun-1} and \eqref{eq:Dun-2} we have
$$
\int_{\R^d} |\Delta w_n|^2  - a^* \int_{\R^d} |w_n|^q \to 0.
$$
For the right side of \eqref{eq:final-1}, note that 
$$
1=\|w_n\|_{L^2}^2= \|w_n(.+x_n)\|_{L^2}^2 =\|w\|_{L^2}^2 + \|w_n(.+x_n)-w\|_{L^2}^2,
$$
since $w_n(.+x_n)\to w$ weakly in $L^2(\R^d)$. Therefore, $\|w\|_{L^2}\le 1$ and $ \|w_n(.+x_n)-w\|_{L^2}^2\le 1$ for $n$ sufficiently large (here we have known that $w\not \equiv 0$). Therefore, by the Gagliardo--Nirenberg inequality \eqref{eq:inte-ineq} to get
$$
\int_{\R^d} |\Delta w|^2 - a^* \int_{\R^d} |w|^q \ge a^* \Big(\|w\|_{L^2}^{2-q}  -1\Big) \int_{\R^d} |w|^q 
$$
and
$$
\liminf_{n\to \infty} \Big( \int_{\R^d} |\Delta (w_n(.+x_n)-w)|^2 - a^* \int_{\R^d} |w_n(.+x_n)-w|^q \Big) \ge 0.
$$
Thus \eqref{eq:final-1} implies that
$$
a^* \Big(\|w\|_{L^2}^{2-q}  -1\Big) \int_{\R^d} |w|^q \le 0.
$$
Since $w\not\equiv 0$, $\|w\|_{L^2}\le 1$ and $q>2$, we conclude that $\|w\|_{L^2}=1$. Thus $w_n(.+x_n)\to w$ strongly in $L^p(\R^d)$ for all $p\in [2,4^*)$. In particular, $w_n(.+x_n)\to w$ strongly in $L^q(\R^d)$, and hence \eqref{eq:final-1} implies that
\begin{align*} 
0\ge \Big( \int_{\R^d} |\Delta w|^2 - a^* \int_{\R^d} |w|^q \Big) +  \int_{\R^d} |\Delta (w_n(.+x_n)-w)|^2 +  o(1)_{n\to \infty}.
\end{align*}
Thus 
$$
\int_{\R^d} |\Delta w|^2 - a^* \int_{\R^d} |w|^q =0,
$$
which means that $w$ is a minimizer for the Gagliardo--Nirenberg inequality \eqref{eq:inte-ineq}, and
$$
\int_{\R^d} |\Delta (w_n(.+x_n)-w)|^2 \to 0,
$$ 
which implies that $w_n(.+x_n)\to w$ strongly in $H^2(\R^d)$. In summary,
$$ \eps_n^{d/2} u_n(\eps_n (x+x_n))= w_n(x+x_n)\to w(x)\quad \text{strongly in~} H^2(\R^d).$$
The proof is complete.
\end{proof}


\begin{thebibliography}{11}




\bibitem{BKS-00} M. Ben-Artzi, H. Koch, J.C. Saut, Dispersion estimates for fourth order
Schr\"odinger equations, C. R. Acad. Sci. Paris Ser. I Math. 330 (2000), no. 2, p. 87--92. 


\bibitem{BL-16} T. Boulenger, E. Lenzmann, Blowup for Biharmonic NLS, Annales Scientifiques de l'\'Ecole Normale Sup\'erieure (to appear), arXiv:1503.01741. 

\bibitem{BreLie-83} H. Brezis and E. Lieb, A relation between pointwise convergence of functions and convergence
of functionals, Proc. Amer. Math. Soc. 88 (1983), p. 486--490.

\bibitem{DL-07} Y. Deng, Y. Li, Exponential decay of the solutions for nonlinear biharmonic equations. Commun. Contemp. Math. 9 (2007), no. 5, p. 753--768. 


\bibitem{DenGuoLu-15} Y. Deng, Y. Guo, L. Lu, On the collapse and concentration of Bose--Einstein condensates with inhomogeneous attractive interactions, Calc. Var. Partial Differential Equations 54 (2015), pp. 99--118.


\bibitem{FIP-02} G. Fibich, B. Ilan, G. Papanicolaou, Self-focusing with fourth-order dispersion, SIAM
J. Appl. Math. 62 (2002), no. 4, p. 1437--1462.


\bibitem{GidNiNir-81} B. Gidas, W.M. Ni, L. Nirenberg, Symmetry of positive solutions of nonlinear elliptic
equations in $\mathbb{R}^n$, Mathematical analysis and applications. Part A, Adv. in Math. Suppl. Stud. Vol. 7, Academic Press, New York, 369--402 (1981).
%

\bibitem{FLL-86} J. Fr\"ohlich, E. H. Lieb, M. Loss, Stability of Coulomb systems with magnetic fields. I. The one-electron
atom, Commun. Math. Phys. 104 (1986), no. 2, p. 251--270


%
\bibitem{GuoSei-14} Y. Guo and R. Seiringer, On the mass concentration for Bose-Einstein condensates with attractive
interactions, Lett. Math. Phys., 104 (2014), pp. 141--156.
%
\bibitem{GuoZenZho-16} Y.J. Guo, X.Y. Zeng, H.S. Zhou, Energy estimates and symmetry breaking in attractive Bose--Einstein condensates with ring-shaped potentials, Ann. Inst. Henri Poincaré 33 (2016), pp. 809--828.
%
%
%
%\bibitem{GuoWangZengZhou-15} Y. Guo, Z.Q. Wang, X. Zeng, H.S. Zhou, Properties of ground states of attractive Gross-Pitaevskii equations with multi-well potentials, arXiv:1502.01839 (2015)
%
%



\bibitem{Karpman-96} V. I. Karpman, Stabilization of soliton instabilities by higher-order dispersion: Fourth-order nonlinear
Schr\"odinger-type equations, Phys. Rev. E 53 (1996), p. 1336--1339.


%
%\bibitem{MclSer-87} K. McLeod, J. Serrin, Uniqueness of positive radial solutions of $\Delta u + f(u) = 0$ in $\mathbb{R}^n$, Arch. Rational Mech. Anal. 99 (1987), p. 115--145.
%
%\bibitem{MorObe-06} O. Morsch, M. Oberthaler, Dynamics of Bose--Einstein condensates in optical lattices, Rev. Mod. Phys. 78 (2006), p. 179--215.
%
%
%  
%\bibitem{LewNamRou-15}  M. Lewin, P.T. Nam, N. Rougerie, The mean-field approximation and the nonlinear Schrödinger functional for trapped Bose gases. Trans. Amer. Math. Soc. 369 (2016), 6131--6157.
%
%%\bibitem{LieLos-01} E. Lieb, M. Loss, Analysis (Graduate Studies in Mathematics), 2nd ed. (2001). 
%
%

\bibitem{Lieb-83} E. H. Lieb, On the lowest eigenvalue of the Laplacian for the intersection of two domains, Invent. Math.
74 (1983), p. 441--448.


%\bibitem{Lions-84} P.L. Lions, The concentration-compactness principle in the calculus of variations. The locally compact case: Part 1, Ann. Inst. Henri Poincar\'e 1 (1984), p. 109--145.
%
%\bibitem{Lions-84b} P.L. Lions, The concentration-compactness principle in the calculus of variations. The locally compact case: Part 2, Ann. Inst. Henri Poincar\'e 1 (1984), p. 223--283.
%

\bibitem{Nirenberg-59} L. Nirenberg, On elliptic partial differential equations, Ann. Scuola Norm. Sup. Pisa 13 (1959), no. 3, p. 115--162.

\bibitem{Pausader-09} B. Pausader, The cubic fourth-order Schr\"odinger equation, J. Funct. Anal. 256 (2009), no. 8, p. 2473--
2517. 


%\bibitem{Phan-17a} T.V. Phan, Blow-up profile of Bose-Einstein condensate with singular potentials, Preprint 2017.
%
%\bibitem{SacStoHul-98} C.A. Sackett, H.T.C. Stoof, R.G. Hulet, Growth and Collapse of a Bose--Einstein Condensate
%with Attractive Interactions, Phys. Rev. Lett. 80 (1998), p. 2031.
%
%\bibitem{Weinstein-83} M. I. Weinstein, Nonlinear Schr\"odinger equations and sharp interpolation estimates, Comm. Math.
%Phys., 87 (1983), pp. 567--576.
%
%\bibitem {WanZha-16} Q. Wang, D. Zhao, Existence and mass concentration of 2D attractive Bose--Einstein condensates with periodic potentials, J. Differential Equation (2016), http://dx.doi.org/10.1016/j.jde.2016.11.004
%
\end{thebibliography}
\end{document}